\documentclass[runningheads]{llncs}

\usepackage[T1]{fontenc}
\usepackage[utf8]{inputenc}
\usepackage{amsmath,amssymb,mathtools}
\usepackage[hidelinks]{hyperref}
\usepackage{xcolor}
\usepackage[normalem]{ulem}

\newif\ifshowchanges
\showchangestrue
\definecolor{reviewpurple}{RGB}{117,31,150}

\usepackage{etoolbox}

\AtEndEnvironment{proof}{\qed}

\newcommand{\negl}{\mathrm{negl}}
\newcommand{\QMA}{\textsf{QMA}}
\newcommand{\QCIP}{\textsf{QCIP}}
\newcommand{\QPT}{\textsf{QPT}}

\newcommand{\val}{\mathrm{val}}

\newcommand{\QCMA}{\textsf{QCMA}}
\newcommand{\Acc}{\operatorname{Acc}}
\newcommand{\E}{\mathbb{E}}

\newcommand{\cU}{\mathcal{U}}
\newcommand{\cT}{\mathcal{T}}
\newcommand{\cO}{\mathcal{O}}

\newcommand{\id}{\mathbb{I}}
\newcommand{\ket}[1]{\lvert #1\rangle}
\newcommand{\bra}[1]{\langle #1\rvert}

\newcommand{\proj}[1]{\lvert #1\rangle\!\langle #1\rvert}
\newcommand{\norm}[1]{\left\lVert #1\right\rVert}

\title{Quantum State Search with Classical Interactive Advice}

\author{Mohammed Barhoush}
\authorrunning{M. Barhoush}
\institute{IonQ Inc.}

\begin{document}
\setlength{\emergencystretch}{2em}
\maketitle

\begin{abstract}
We study whether a quantum verifier can find a hidden quantum state
by exchanging classical messages with an unbounded prover who knows
the state. The verifier has access to a reflection oracle about
the target.
For a Haar-random $n$-qubit target, we prove that achieving constant
average fidelity with the target requires
\[
                   (q+1)^2(B+1)=\Omega(2^n),
\]
where $q$ is the number of verifier queries and $B$ is the total
prover-to-verifier communication. The bound allows arbitrary
classical interaction and unrestricted computation between queries. Our proof extends the lower bound of Aaronson and Kuperberg (ToC'07) from
a single classical proof to an arbitrary classical interaction.

As a result, we obtain a quantum-oracle separation of
$\QMA$ from $\textsf{QCIP}$. {In particular}, we give a unitary oracle $\mathcal{U}$ such that there exists a language $L\in \textsf{QMA}^{\mathcal{U}}$ but $L\notin \QCIP^{\mathcal{U}}$. Our separation holds even for
$\textsf{UniqueQMA}$, the restriction of $\textsf{QMA}$ to verification
with a unique quantum witness, which is known to be strictly weaker
than $\textsf{QMA}$ in the quantum black-box setting.

We further give an oracle construction for a notion known as 
$\textsf{QMA}$--$\textsf{QCIP}[2]$ gap problem. This gives average case hardness for a $\QMA$ language under classical access to a $\textsf{QCIP}[2]$ decision oracle and is used in establishing negative results on witness encryption and quantum indistinguishability obfuscation.
\end{abstract}

\newpage

\section{Introduction}
\label{sec:intro}

Unstructured search is one of the basic problems in computer science.  For a search space of
size $2^n$, Grover's algorithm uses $O(2^{n/2})$ queries, and this dependence
is optimal~\cite{BennettBernsteinBrassardVazirani1997,Grover1996}.
A prover who knows the marked item can remove this query requirement by
sending its $n$-bit address to the verifier. Does classical interaction
provide a similar advantage when the hidden object is a quantum state?

Aaronson and Kuperberg~\cite{AaronsonKuperberg2007} studied this question
for a single classical message. They used an oracle that marks an unknown
$n$-qubit state $|\psi\rangle$ by applying the reflection
\[
                    R_\psi=\id-2|\psi\rangle\langle\psi|.
\]
They proved that polynomial classical advice is not sufficient to find the hidden state and distinguish  this oracle from the identity, given only polynomial queries. On the other hand, quantum advice easily allows distinguishing this oracle from the identity using $\psi$ as advice. 

{These findings have important implications for} the relationship between the classes $\QMA$ and $\textsf{QCMA}$. Informally, both these classes describe problems whose
solutions can be verified by a polynomial-time
quantum verifier, using a quantum witness in the case of 
$\textsf{QMA}$ and a classical witness in the case of $\textsf{QCMA}$.
The relation between these classes captures a fundamental question
in quantum complexity theory: can quantum witnesses enable more
powerful verification than classical witnesses? 

{Aaronson and Kuperberg's result shows} that a classical witness cannot do what a quantum witness can do in this oracle model, thus demonstrating the first oracle separation between $\QMA$ and $\textsf{QCMA}$. 

We ask whether the same advantage holds against classical \emph{interactive} advice. In our setting, the verifier {interacts with an unbounded prover} that knows $\psi$ and tries to help the verifier find the hidden state through classical interaction.  The verifier may retain quantum memory and interleave
classical messages with oracle queries. {Such a protocol cannot simply
be replaced by one short classical advice: there is not a necessarily  ``optimal'' transcript of interactive classical messages that can be sent in one-shot to the verifier since the success of a transcript may depend on the internal state of the verifier and quantum computations cannot be derandomized.}

Our main result extends the Aaronson--Kuperberg lower bound to arbitrary
classical interaction. 

\begin{theorem}[Hidden-state search with classical interaction, informal]
\label{thm:intro-main}
Let $|\psi\rangle$ be a Haar-random $n$-qubit state. If the verifier
makes $q$ oracle queries to the oracle $R_\psi$ and the prover sends at most $B$ classical
bits in total, then producing an output with constant average fidelity with
$|\psi\rangle$ requires
\[
                   (q+1)^2(B+1)=\Omega(2^n).
\]
\end{theorem}

Fidelity means $\langle\psi|\rho|\psi\rangle$ for the verifier's output
$\rho$, averaged over the hidden state and the protocol's randomness. For polynomial prover-to-verifier communication, the theorem gives an
exponential query requirement. For $2n\leq B=o(2^n)$, it matches the
one-message upper bound of Aaronson and
Kuperberg~\cite{AaronsonKuperberg2007}.
Classical interaction therefore does not aid in performing unstructured search for a hidden Haar random state.

The Aaronson--Kuperberg separation actually separates $\textsf{UniqueQMA}$ from $\QCMA$, since the $\QMA$ language built in their oracle model has one unique witness up to global phase. The class $\textsf{UniqueQMA}$ restricts $\textsf{QMA}$ to
yes-instances with a one-dimensional accepting witness subspace:
a designated witness is accepted with high probability, while
every orthogonal witness is accepted with low
probability~\cite{JainKerenidisKuperbergSanthaSattathZhang2012}. 

Our oracle construction generalizes the Aaronson--Kuperberg separation by extending the $\QCMA$ to the class $\QCIP$, which is the complexity class with a quantum polynomial-time verifier and a computationally--unbounded prover that can interact with each other through classical messages. In particular, as a direct consequence of Theorem \ref{thm:intro-main} we obtain the following result.
\begin{theorem}
\label{thm:implication-unique}
There is a quantum oracle $\mathcal U$ such that
\[
    \textsf{UniqueQMA}^{\mathcal U}
    \not\subseteq
    \textsf{QCIP}^{\mathcal U}.
\]
\end{theorem}
 The class $\QCIP$ was introduced by Buhrman, Le Gall, and
Weggemans~\cite{BuhrmanLeGallWeggemans2024} and  was further studied by
Grewal and Kretschmer~\cite{GrewalKretschmer2025}. Notice that $\QCIP[1]=\QCMA$\footnote{For an integer $k$, we let $\QCIP[k]$ denote the $\QCIP$ class restricted to $k$ classical messages. }. This separation  does not assert an
unrelativized separation, since polynomial-round $\textsf{QCIP}$
contains $\textsf{QMA}$ in the plain model~\cite{BuhrmanLeGallWeggemans2024}. {It is also consistent with computationally sound classical
verification of quantum computations~\cite{Mahadev2022}, whose soundness under the learning-with-errors assumption applies only against efficient cheating provers.}

In concurrent and independent work, \c{C}akan~\cite{Cakan2026} separates
$\textsf{QMA}$ from polynomial-round $\textsf{QCIP}$ using a different approach than our own. \c{C}akan's separation has the advantage of being relative to a classical oracle which is a more restricted model. In comparison, our approach provides a simpler and perhaps more intuitive separation. Moreover, our result separates \textsf{UniqueQMA} rather than $\QMA$, and thus provides a more fine-tuned separation. \c{C}akan's  separation does not provide a unique witness to the constructed $\QMA$ language so it does not separate \textsf{UniqueQMA} from $\QCIP$.  Note that uniqueness is a substantive restriction: there exists black-box separations between $\textsf{QMA}$ and
$\textsf{UniqueQMA}$~\cite{AnshuHaferkampHwangNguyen2026}.

\paragraph{An average-case gap problem.}
Next, we use our lower bound to build a primitive called a $\QMA$--$\QCIP[2]$ gap problem, introduced by Barhoush, Morimae, and
Movassagh~\cite{BarhoushMorimaeMovassagh2026}. {Assuming this gap problem, they present a restricted black-box separation between witness encryption with classical ciphertexts and any falsifiable assumption.} This class of assumptions essentially includes any assumption where the challenger can efficiently check whether the adversary has won the game and includes one-way functions and LWE. Their barrier also implies a separation between  classically encoded quantum indistinguishability obfuscation and falsifiable assumptions. {We provide an oracle construction} of this gap problem which supports these separations. 

\begin{theorem}[Oracle gap problem, informal]
\label{thm:intro-gap}
There exist a quantum unitary oracle $\mathcal T$ and a language
$G\in\textsf{QMA}^{\mathcal T}$ with yes- and no-instances efficiently
samplable relative to $\mathcal T$. Nevertheless, the classical yes- and no-instance
distributions remain computationally indistinguishable to quantum polynomial-time
algorithms with access to $\mathcal T$ {and classical access to a \mbox{$\textsf{QCIP}[2]^{\mathcal T}$} decision oracle.}
\end{theorem}

Barhoush, Morimae, Nishimaki, and
Yamakawa~\cite{BarhoushMorimaeNishimakiYamakawa2026} introduce a similar notion of
$\textsf{QMA}$--$\textsf{QCMA}$ gap assumption, which they build relative to a quantum oracle.
Assuming such a gap, they obtain a quantum black-box barrier for
non-interactive classical verification of quantum computation from falsifiable assumptions. We adapt their oracle construction to permit
$\textsf{QCIP}[2]$ decision access.

We will now discuss how our results relate to various areas in quantum cryptography.

\paragraph{State preparation with verification access.}
The task of transmitting a quantum state using classical interaction is known as remote state preparation. This area has well-established communication--entanglement
tradeoffs~\cite{BennettEtAl2001,BennettEtAl2005}, including lower bounds~\cite{KunduLalonde2026}. These models do not provide the verifier with a reflection oracle
for the state to be prepared, which allows the ability to verify a state. Our contribution is to include this
additional resource and quantify its tradeoff with classical
communication. Note that verification ability can help restore a state when a sufficiently large
subsystem is already available~\cite{FarhiEtAl2010}. The task of preparing state given the ability to verify is also closely related to publicly verifiable quantum
money~\cite{AaronsonChristiano2013,Aaronson2009}.

\paragraph{Comparison with quantum no-telegraphing and quantum fire.}
Nehoran and Zhandry~\cite{NehoranZhandry2024} establish quantum
no-telegraphing even when the sender is computationally unbounded:
a receiver with bounded oracle access cannot reconstruct the state
from a polynomial-length classical description. Their result concerns
one-way communication, whereas our lower bound permits arbitrary
classical interaction.

\c{C}akan, Goyal, and Shmueli~\cite{CakanGoyalShmueli2026} establish
interactive untelegraphability relative to a classical oracle.
Their model allows unbounded computation between oracle queries
and even unrestricted classical communication, but bounds the
oracle queries of \emph{both} the sender and the receiver.
This restriction on the sender is essential when comparing their
result with ours: their security guarantee does not cover classical
messages obtained through unrestricted oracle computation. In our setting, the prover has complete
knowledge of the hidden state \( |\psi\rangle \) and may perform
arbitrary computation and oracle queries. We bound only the total
prover-to-verifier communication and the verifier's oracle queries.
Thus, our lower bound does not rely on the prover's inability to
learn the state or obtain a useful description of it. 

Casper, Nehoran, and Sattath~\cite{CasperNehoranSattath2026}
construct conversable quantum fire in the standard model under
cryptographic assumptions. Their states resist one-way classical
transfer by quantum polynomial-time parties but admit transfer
through classical interaction.


\subsection{Acknowledgments}
The concepts and proof ideas were mainly generated by the authors. We used ChatGPT 6 Astra to help write the proofs, improve presentation and check for mistakes. The authors take responsibility for any mistakes or issues.

\subsection{Technical overview}

We now give a simplified description of the proof of our separation.

The unitary oracle and $\textsf{UniqueQMA}$ language we construct is the same as the Aaronson--Kuperberg separation. Their oracle hides an unknown Haar random \(n\)-qubit pure state \(\ket{\psi}\) through the reflection
\[
        R_\psi=\id-2\proj{\psi}.
\]
The promise is that the oracle is either \(R_\psi\), for some hidden state
\(\ket{\psi}\), or the identity.  A \textsf{UniqueQMA} prover can send
\(\ket{\psi}\) itself, after which one controlled oracle query distinguishes
the two cases perfectly.  Aaronson and Kuperberg proved that any verifier that
distinguishes \(R_\psi\) from \(\id\) with bounded error on average for a Haar random \(n\)-qubit
state \(\psi\), even when given an \(m\)-bit classical witness chosen as an
arbitrary function of \(\psi\), must make
\(\Omega(\sqrt{2^n/(m+1)})\) queries. 

We need to generalize this argument to $\QCIP$ provers. We first show that the prover in a $\QCIP$ protocol can be made deterministic. In particular, for each hidden state $\psi$, we fix one optimal prover that may depend arbitrarily on the hidden state. The idea is that given any transcript of messages and the state $\psi$, there must exist a response by the prover that maximizes the probability of the verifier accepting. We fix such a maximizing response for each transcript to obtain a deterministic prover. 

\paragraph{The transcript-information lemma.}
{Set $D=2^n$, and let $\mu_D$ be the Haar probability measure on
pure states in $\mathbb C^D$. For each $\psi$, let $P_\psi$ be the
deterministic prover fixed above. A public transcript $h$ is the ordered list of
classical messages exchanged so far. Let $C(h)$ be the set of hidden states whose provers send all
the replies recorded in $h$, when given the verifier messages recorded
there. For a positive-probability transcript $h$, define}
\[
                     L(h):=\ln\frac1{\mu_D(C(h))}.
\]
{The quantity $\mu_D(C(h))$ is the probability
that a Haar-random state belongs to $C(h)$. We now consider the experiment in which the {verifier's access to the reflection oracle} is replaced with access to an
identity oracle. In this case, the verifier only has knowledge of the state $\psi$ through the messages in the transcript $h$. In other words, the probability of generating the verifier
messages in $h$ is the same for every state in $C(h)$. Thus, conditioned
on $h$, the hidden state has Haar measure restricted to $C(h)$, and the
verifier's conditional private state depends only on $h$.}
Thus \(L(h)\) measures how much the transcript has narrowed the possible
hidden states relative to the verifier.  Initially,
\[
 C(\varnothing)=\mathbb{CP}^{D-1}
 \qquad\text{and}\qquad
 L(\varnothing)=0.
\]

A verifier message does not change this quantity.  Indeed, conditioned on the
earlier transcript \(h\), the verifier's state is independent of the particular
hidden state in \(C(h)\).  Therefore, if the verifier next sends \(a\), then
\[
                         C(ha)=C(h)
                         \quad\text{and}\quad
                         L(ha)=L(h).
\]

Now suppose that the prover sends a reply \(z\).  Since \(P_\psi\) is
deterministic, this reply restricts the compatible set to
\[
                  C(hz)=\{\psi\in C(h):P_\psi(h)=z\}.
\]
The sets \(C(hz)\), as \(z\) ranges over the possible replies, partition
\(C(h)\).  Define
\[
                  p_z:=\frac{\mu_D(C(hz))}{\mu_D(C(h))}.
\]
Note that \(p_z\) is exactly the conditional probability that the
prover sends \(z\), given the earlier transcript \(h\).

If reply \(z\) is observed, the increase in \(L\) is
\[
 \begin{aligned}
 L(hz)-L(h)
 &=\ln\frac1{\mu_D(C(hz))}
   -\ln\frac1{\mu_D(C(h))} \\
 &=\ln\frac{\mu_D(C(h))}{\mu_D(C(hz))}
 =\ln\frac1{p_z}.
 \end{aligned}
\]
Therefore, averaging over the possible replies, the expected
increase is
\[
 \E[L(hZ)-L(h)\mid h]
 =\sum_z p_z\ln\frac1{p_z}.
\]
Here the expectation is over the distribution of $\psi$. The right-hand side is the entropy of the prover's reply, measured using
natural logarithms.  If the reply contains \(b\) bits, it has at most \(2^b\)
possible values, and hence
\[
                  \sum_z p_z\ln\frac1{p_z}
                  \leq\ln(2^b)=b\ln2.
\]
Thus, a \(b\)-bit prover message can increase the expected value of \(L\) by
at most \(b\ln2\), while verifier messages do not increase it at all.

Finally, let \(B_t\) be the total number of prover bits sent before query
\(t\) and let $T_t$ denote the public transcript immediately
before that query.  Summing the preceding bound over all prover messages before that
query gives
\[
 \E\!\left[L(T_t)\right]
 =
 \E\!\left[\ln\frac1{\mu_D(C(T_t))}\right]
 \leq B_t\ln2.
 \tag{II}
\]
This expectation correctly weights each transcript by its probability:
transcripts with very small compatible sets may have large \(L(h)\), but
they are correspondingly unlikely to occur.

\paragraph{From information to query complexity.}
The bound on $\mathbb{E}[L]$ allows us to bound the information the verifier has before each oracle query. In particular, it allows us to show that before any query, the query
register has expected overlap at most $O((B+1)/D)$ with the
hidden state. Replacing that identity query by a reflection
changes the state by at most twice the square root of this
overlap in trace distance i.e. $O(\sqrt{(B+1)/D})$.
We now replace all $q$ queries one at a time until we return to the standard experiment where all queries by the verifier are answered with the reflection oracle. The triangle inequality gives a
total expected change of at most
\[
    O\!\left(q\sqrt{\frac{B+1}{D}}\right).
\]
This is where $q$ enters: each query contributes to the total
change. Constant distinguishing advantage consequently requires $(q+1)^2(B+1)=\Omega(D)$.

\paragraph{The $\QMA$--$\QCIP[2]$ gap problem.}
For the gap problem, we follow a similar construction to that of Barhoush,
Morimae, Nishimaki, and Yamakawa \cite{BarhoushMorimaeNishimakiYamakawa2026}.

At length \(n\), choose a random sparse set \(L_n\) of yes-labels and
an independent Haar-random state \(\ket{\psi_x}\) for each label \(x\).
The construction has three oracle components:
\begin{itemize}
    \item \textbf{Testing oracle.}
    Given a label \(x\) and a quantum state, this oracle applies the
    reflection about \(\ket{\psi_x}\) if \(x\in L_n\), and the identity
    otherwise. Thus, a verifier given \(\ket{\psi_x}\) as a quantum
    witness can test whether \(x\) is a yes-instance using one
    controlled query.

    \item \textbf{Yes-sampling oracle.}
    Starting from a blank state, this oracle prepares a uniform
    superposition of yes-labels, each accompanied by a requested
    number of copies of its corresponding witness.
    Measuring the label register gives a uniformly random
    \(x\in L_n\), together with copies of \(\ket{\psi_x}\).

    \item \textbf{No-sampling oracle.}
    Starting from a blank state, this oracle prepares a uniform
    superposition of labels outside \(L_n\).
    Measuring the label register gives a uniformly random no-instance.
\end{itemize}

To demonstrate the existence of a $\QMA$--$\QCIP$ gap problem relative to these oracles, we need to establish that no QPT adversary can distinguish a randomly sampled yes-instance from a randomly sampled no-instance, even given access to these oracles as well as a $\QCIP[2]$ oracle. 

Suppose an adversary receives an instance $y\in L$. Since the part of the oracle generating yes instances only generates them in a uniform superposition, it is easy to carry out a  hybrid argument where $y$ is deleted from the yes-sampling oracle and replaced into the no-sampling oracle. These changes only have negligible difference on the output of the oracles given that the state generated are in uniform superposition over exponentially many elements.

However, the adversary still has access to a $\QCIP$ oracle, which seems stronger than a $\QCIP$ verifier: the adversary can adaptively query the oracle to decide $\QCIP$ instances. Nevertheless, a change of perspective makes this situation less daunting. The adversary's interaction with the oracle resembles that of a $\QCIP$ verifier with a prover, since both involve classical interaction. By utilizing this idea, we show that our lower-bound arguments for $\QCIP$ verifiers also apply to the adversary, establishing that it cannot find the hidden state associated with $y$. This completes the indistinguishability argument for the $\QMA$--$\QCIP$ gap problem.

Finally, a short second-moment
argument followed by Borel--Cantelli fixes one base oracle that work for all efficient distinguishers.  The formal construction
and proof appear in Section~\ref{sec:gap}.

\section{Model and Preliminaries}
\label{sec:model}

All logarithms written as \(\ln\) are natural logarithms.  For density
operators \(\rho\) and \(\sigma\), let
\[
   \mathrm D(\rho,\sigma)=\frac12\norm{\rho-\sigma}_1
\]
denote trace distance.  Trace distance is contractive under quantum channels
and bounds the difference between the probabilities of every measurement
outcome.

Let \(\mu_D\) be Haar measure on the pure states in \(\mathbb C^D\), or
equivalently on \(\mathbb{CP}^{D-1}\).  We write
\(\Psi\leftarrow\mu_D\) for a Haar-random pure state.

A \(\QPT\) algorithm means a uniform
polynomial-time quantum oracle algorithm.  

\subsection{The marked-state reflection oracle}

For every unit vector \(\ket{\psi}\in\mathbb C^D\), define
\[
   R_\psi=\id-2\proj{\psi}.
\]
The two cases of the promise problem are
\[
  \textsf{marked}: R_\psi\text{ for some }\ket{\psi},
  \qquad
  \textsf{unmarked}: \id.
\]
As in Aaronson and Kuperberg~\cite{AaronsonKuperberg2007}, controlled oracle
calls are allowed.  This only strengthens the verifier.  

\subsection{Classical-message quantum interactive proofs}

We make the communication model explicit.  The verifier and prover start in
a fixed product state; in particular, they have no shared classical or
quantum setup and no shared entanglement.  They exchange only classical
strings.  Formally, every outgoing message register is measured, or
equivalently completely dephased, in the computational basis before it is
transmitted.  The verifier and prover may retain arbitrary private quantum
registers.  A \emph{public transcript} is the round-indexed list of all
transmitted classical messages.  It does not include either party's private
measurement outcomes.

Messages alternate, and the prover sends the last counted message.  More
explicitly, for \(r=2p-1\) the order is
\[
  Z_1\longrightarrow A_1\longrightarrow Z_2
  \longrightarrow\cdots\longrightarrow A_{p-1}\longrightarrow Z_p,
\]
while for \(r=2p\) it is
\[
  A_1\longrightarrow Z_1\longrightarrow A_2
  \longrightarrow\cdots\longrightarrow A_p\longrightarrow Z_p.
\]
Here \(A_i\) is sent by the verifier and \(Z_i\) by the prover.  Thus
\(\QCIP[1]=\QCMA\), \(\QCIP[2]\) is verifier-first, and
\(\QCIP[3]\) has order \(P\to V\to P\).  Padding by ignored messages makes
equivalent conventions harmless. Following Buhrman, Le Gall, and
Weggemans, we use \(\QCIP\) without a bracketed parameter for the union of
polynomial-time protocols having at most \(n^d\) classical messages, over all
fixed integers \(d\geq1\).  Their total communication is also polynomial.

On a fixed input, a \emph{deterministic classical prover} simply specifies
one classical reply after every possible public transcript.  We write this reply as \(P(h)\), where \(h\) is the
transcript.  No efficiency or uniformity condition is imposed on \(P\).
When the oracle is \(W\), let \(\Acc_W(V,P)\) be the verifier's acceptance
probability when interacting with \(P\), and define
\[
   \val_W(V)=\max_P \Acc_W(V,P),
\]
where the maximum is over deterministic classical provers.

\begin{lemma}[Deterministic classical provers suffice]
\label{lem:deterministic}
Fix an oracle \(W\) and a protocol with classical communication and no
shared initial state.  The maximum acceptance probability over all quantum
provers is attained by a deterministic classical prover.
\end{lemma}

\begin{proof}
Fix any quantum prover.  Conditioned on a public transcript, the parties'
private states are a product
state.  Indeed, they start with no shared state; local computations preserve
this property; and conditioning on a transmitted classical message changes
only the sender's private state.  Thus, once the transcript is fixed, the
verifier's conditional state does not depend on the prover's private quantum
memory.

We now fix the prover's replies, starting with its last message and working
backwards.  At a transcript where the prover sends its last message, there is a probability distribution over the possible
classical replies.  Its acceptance probability is therefore an average of
the acceptance probabilities obtained by fixing the reply.  Some fixed reply
is at least as good as this average, so choose such a reply for that
transcript.  At a transcript of probability zero, choose any reply.  Do this
for every transcript at the last prover turn.  Then
repeat the same argument at the preceding prover turn, with all later replies
already fixed, and continue backwards.  This produces a deterministic reply
for every public transcript, and the resulting deterministic classical prover
accepts at least as often as the original quantum prover.  Finally, the
protocol has finitely many turns and finite message registers, so there are
only finitely many deterministic classical provers.  One of them therefore
attains the largest acceptance probability.
\end{proof}

When \(W=R_\psi\), we allow the deterministic prover \(P\) to depend
arbitrarily on \(\psi\).  This already grants the prover everything it could
learn from unlimited oracle access.

 To define
oracle-query slots unambiguously, divide the verifier's computation
into phases separated by transmitted messages.  Within each phase, pad every
adaptive query path to the phase's maximum number of queries, using
oracle calls controlled on a zero flag.  Every execution then has the same
chronologically ordered list of query slots.  For slot \(t\), let \(B_t\)
be the sum of the padded lengths of the prover messages in earlier
communication slots.  This is a deterministic bound, independent of the
branch taken.  Let \(B\) be the total padded prover-to-verifier
communication in the protocol.  Private measurements may be deferred, while
the measurements that create transmitted classical messages are kept
explicit.

\subsection{Relativized model}

A quantum oracle is a sequence \(\cU=(U_1,U_2,\ldots)\), where \(U_n\) is a
unitary on a specified polynomial number of qubits.  We use the
length-respecting convention of Aaronson and Kuperberg: on input \(0^n\), a
machine queries only \(U_n\).  This makes the diagonal construction
transparent.

\subsection{Geometric Lemma}
\label{sec:haar}

We state an important geometric lemma used in
Aaronson--Kuperberg separation \cite{AaronsonKuperberg2007}.

\begin{definition}
A probability measure \(\nu\) on \(\mathbb{CP}^{D-1}\) is
\emph{\(p\)-uniform}, for \(0<p\leq1\), if
\(p\nu(G)\leq\mu_D(G)\) for every measurable set \(G\).
Equivalently, we write \(p\nu\leq\mu_D\) as an inequality of measures.
\end{definition}

\begin{lemma}[Haar overlap under conditioning (Lemma 3.2 in \cite{AaronsonKuperberg2007})]
\label{lem:puniform}
Let \(\nu\) be \(p\)-uniform and let \(\rho\) be a density operator on
\(\mathbb C^D\), where \(D\geq2\).  Then
\[
  \E_{\Psi\leftarrow\nu}
    [\bra{\Psi}\rho\ket{\Psi}]
  \leq \frac{1+\ln(1/p)}{D}.
\]
\end{lemma}

\section{What a Classical Transcript Reveals}
\label{sec:transcript}

We first describe the experiment used throughout this section.  Fix a
verifier \(V\) with \(q\) padded query slots, fix \(t\in[q]\), and fix a
measurable choice \(\psi\mapsto P_\psi\), for
\(\psi\in\mathbb{CP}^{D-1}\), of deterministic classical provers.
The measurable choice needed for the lower bound is constructed explicitly
in Theorem~\ref{thm:value}.  Sample
\(\Psi\leftarrow\mu_D\), let the prover be \(P_\Psi\), and replace
every verifier oracle call before slot \(t\) by the identity.  The prover is
already represented by the fixed prover \(P_\Psi\), so it makes no
oracle calls inside this experiment.  We call this the
\emph{identity-prefix execution}.  Let \(T_t\) be its public classical
transcript immediately before query \(t\).  The randomness in \(T_t\)
comes from the Haar state and from the verifier's measurements and private
randomness.  When we say that the verifier's state is fixed by a transcript,
we mean that its conditional density operator is fixed; it may be mixed
because private outcomes are averaged.

For a possible transcript \(h\), define its \emph{compatible set}
\[
 C(h)=
 \{\psi\in\mathbb{CP}^{D-1}:
   \text{$P_\psi$ sends every prover message recorded in $h$}\}.
\]
Thus \(C(h)\) contains exactly the hidden states consistent with the
prover's part of \(h\).  Verifier messages still appear in \(h\), because
they determine which reply the prover sends.  Their own
probabilities impose no further restriction on \(\psi\) once the preceding
transcript is fixed.  Measurability of the choice
\(\psi\mapsto P_\psi\) makes every set \(C(h)\) Haar measurable.

\paragraph{Warm-up: three messages.}
Consider a transcript
\[
       Y\longrightarrow A\longrightarrow Z,
\]
where \(Y,Z\) are prover messages and \(A\) is the verifier's challenge.
Capital letters denote random messages, and lower-case letters denote their
realized values.
After observing \(Y=y\), the compatible states form \(C(y)\).  Once
\(y\) is fixed, the identity-prefix verifier's private state is fixed as
well.  Hence the distribution of \(A\) is the same for every
\(\psi\in C(y)\), and observing \(A=a\) does not shrink the set:
\[
       C(y,a)=C(y).
\]
The final prover message can shrink it:
\[
 C(y,a,z)=\{\psi\in C(y):P_\psi(y,a)=z\}.
\]
Thus \(Y\) and \(Z\) can reveal information about \(\Psi\), but \(A\)
cannot reveal new information once \(Y\) is known.  If query \(t\) occurs
before \(Z\), then only the bits in \(Y\) are charged.


\begin{lemma}[Query overlap after a classical transcript]
\label{lem:transcript-info}
\label{lem:query-overlap}
In the identity-prefix execution, let \(\rho_{t,h}\) be the state of the
oracle-input register before query \(t\), conditioned on transcript
\(T_t=h\). If at most \(B_t\) padded prover bits precede this query, then
\begin{equation}
 \E_{\Psi,T_t}
   [\bra{\Psi}\rho_{t,T_t}\ket{\Psi}]
 \leq \frac{1+B_t\ln2}{D}.
 \label{eq:query-overlap}
\end{equation}
\end{lemma}

\begin{proof}
We first show that, conditioned on any possible transcript
\(h\), the hidden state is Haar distributed over \(C(h)\), and the
verifier's private state depends only on \(h\), not on the particular
\(\psi\in C(h)\). 

Both statements hold at the empty transcript. Now suppose first that the
verifier sends the next message. Given the preceding transcript, its private
state is fixed because all earlier verifier oracle calls are identities.
Therefore the distribution of its message, and its residual
post-measurement state, do not depend on \(\psi\). The compatible set
is unchanged. Similarly, at any point in the protocol, when the verifier sends a message, it does not change the compatible set.  

Suppose instead that the prover sends the next message \(z\). Since
\(P_\psi\) is deterministic, the new compatible set is exactly
\[
 C(hz)=\{\psi\in C(h):P_\psi(h)=z\},
\]
where \(hz\) denotes \(h\) extended by \(z\). Conditioning on this reply
restricts the hidden state to \(C(hz)\). This proves the induction.

We next bound the expected logarithmic inverse mass of the compatible set.
At the empty transcript this quantity is zero, and a verifier message
leaves it unchanged. A padded \(b\)-bit prover message partitions
\(C(h)\) into the sets \(C(hz)\). Let
\[
 p(z\mid h)=\frac{\mu_D(C(hz))}{\mu_D(C(h))}.
\]
Therefore, the expected logarithmic mass of the compatible set is:
\begin{align*}
 \sum_z p(z\mid h)\ln\frac{1}{\mu_D(C(hz))}
 &=\ln\frac{1}{\mu_D(C(h))}
   +\sum_z p(z\mid h)\ln\frac1{p(z\mid h)}\\
 &\leq \ln\frac{1}{\mu_D(C(h))}+b\ln2.
\end{align*}
The inequality holds because a \(b\)-bit message has at most \(2^b\)
possible values, so its conditional entropy is at most \(b\ln2\).
Average over the preceding transcript and repeat for every prover message
before query \(t\). Their padded lengths total at most \(B_t\), giving
\begin{equation}
 \E\!\left[\ln\frac{1}{\mu_D(C(T_t))}\right]
 \leq B_t\ln2.
 \label{eq:transcript-information}
\end{equation}

For each possible transcript \(h\), the query state
\(\rho_{t,h}\) is fixed, whereas the hidden state has Haar measure
conditioned on \(C(h)\). Lemma~\ref{lem:puniform} therefore gives
\[
 \E[\bra{\Psi}\rho_{t,h}\ket{\Psi}\mid T_t=h]
 \leq \frac{1+\ln(1/\mu_D(C(h)))}{D}.
\]
Averaging over \(T_t\) and applying
Eq.~\eqref{eq:transcript-information} proves the claim.
\end{proof}

\section{The General Query Lower Bound}
\label{sec:lowerbound}

We combine Lemma~\ref{lem:query-overlap} with the standard quantum-query
hybrid method~\cite{BennettBernsteinBrassardVazirani1997}.
The following theorem also allows access to a fixed background oracle
independent of \(\Psi\), which we suppress from the notation.

\begin{theorem}[Classical-interaction query bound]
\label{thm:value}
Let \(V\) exchange only classical messages, make at most \(q\) queries to
the distinguished \(D\)-dimensional oracle, and receive at most \(B\)
padded prover bits. Then
\begin{equation}
 \E_\Psi\!\left[
   \max_P\left|
     \Acc_{R_\Psi}(V,P)-\Acc_{\id}(V,P)
   \right|
 \right]
 \leq 2q\sqrt{\frac{1+B\ln2}{D}},
 \label{eq:coarse-value}
\end{equation}
where the maximum is over deterministic classical provers. The same bound
holds for \(\E_\Psi[|\val_{R_\Psi}(V)-\val_{\id}(V)|]\).
\end{theorem}

\begin{proof}
For every \(\psi\), choose a deterministic classical prover \(P_\psi\)
maximizing the absolute difference in Eq.~\eqref{eq:coarse-value}.
At each fixed input length, there are finitely many possible deterministic
response functions on the protocol's transcript tree. Order them and choose
the first maximizer. Since acceptance probabilities for each fixed prover
are continuous in \(\psi\), this choice is Borel measurable.

Use the same prover throughout the hybrid. Only the verifier makes oracle calls as the prover has full knowledge of $\psi$. Order the verifier's oracle calls chronologically. For
\(t=0,\ldots,q\), let \(\Omega_t^\psi\) be the final state of the
verifier and public transcript when the first \(t\) calls use \(\id\)
and the remaining calls use \(R_\psi\). 

To bound the effect of one replacement, purify the state before the query
to \(\ket{\Phi}\), and let \(u\) be the overlap of its oracle-input
register with \(\ket{\psi}\). The two post-query purifications have
inner product \(1-2u\), so their trace distance is
\[
 \sqrt{1-|1-2u|^2}=2\sqrt{u(1-u)}\leq2\sqrt u.
\]
Discarding the purifying register and applying later operations cannot
increase trace distance. For a controlled query, \(u\) is replaced by
the weight on \(\ket{1}\otimes\ket{\psi}\), which is at most \(u\),
so the same bound holds.

Conditioned on transcript \(h\) before query \(t\), let
\(\rho_{t,h}\) be the reduced density operator of the oracle-input
register, obtained by tracing out all other registers; the overlap
defined above is therefore \(u=\langle\psi|\rho_{t,h}|\psi\rangle\). Convexity of trace distance, concavity of the square root, and
Lemma~\ref{lem:query-overlap} give
\begin{align*}
 \E_\Psi[\mathrm D(\Omega_{t-1}^\Psi,\Omega_t^\Psi)]
 &\leq 2\E_{\Psi,T_t}
     \left[\sqrt{\bra{\Psi}\rho_{t,T_t}\ket{\Psi}}\right]\\
 &\leq 2\sqrt{\frac{1+B\ln2}{D}}.
\end{align*}

We can telescope over the \(q\) queries to obtain
Eq.~\eqref{eq:coarse-value} for the chosen maximizing provers.

We now compare the optimal acceptance probabilities. Suppose first
that \(\val_{R_\psi}(V)\geq\val_{\id}(V)\), and choose a deterministic
prover \(P^\star\) that achieves \(\val_{R_\psi}(V)\).
This prover specifies a classical reply to every possible transcript.
We can use exactly the same reply rules when the verifier instead
queries the identity oracle. Since \(\val_{\id}(V)\) is the maximum
acceptance probability over all prover strategies, it satisfies
\[
 \val_{\id}(V)\geq \Acc_{\id}(V,P^\star).
\]
Together with
\(\val_{R_\psi}(V)=\Acc_{R_\psi}(V,P^\star)\), this gives
\[
 \val_{R_\psi}(V)-\val_{\id}(V)
 \leq
 \Acc_{R_\psi}(V,P^\star)-\Acc_{\id}(V,P^\star).
\]
Thus, the difference between the optimal values is at most the
difference obtained by keeping this one prover fixed.
If the identity-oracle value is larger, choose a prover optimal
for the identity oracle and reverse the argument. In either case,
\[
 \bigl|\val_{R_\psi}(V)-\val_{\id}(V)\bigr|
 \leq
 \max_P
 \bigl|\Acc_{R_\psi}(V,P)-\Acc_{\id}(V,P)\bigr|.
\]
Averaging over the hidden state and applying
Eq.~\eqref{eq:coarse-value} proves the second part of the theorem.
In particular, a constant average absolute gap between the two
values requires
\[
 q=\Omega\!\left(\sqrt{\frac{D}{B+1}}\right).
\]
\end{proof}

\section{One Oracle Separating Every Message Count}
\label{sec:separation}

For each \(n\geq 1\), choose a bit \(b_n\in\{0,1\}\) and a Haar random $n$-qubit state \(|\psi_n\rangle\).
Define the \(n\)-qubit unitary
\[
 U_n=\id_{2^n}-2b_n|\psi_n\rangle\langle\psi_n|,
\]
where \(\id_{2^n}\) is the identity on \(n\) qubits.
Thus \(U_n=\id_{2^n}\) when \(b_n=0\), while
\(U_n=R_{\psi_n}\) when \(b_n=1\).
The quantum oracle is the family
\(\mathcal U=(U_n)_{n\geq 1}\).
The construction below specifies the bits \(b_n\) and states
\(|\psi_n\rangle\).

For each length \(n\), define
\[
 L_n^{\mathcal U}=
 \begin{cases}
   \{0^n\}, & b_n=1,\\
   \varnothing, & b_n=0.
 \end{cases}
\]
The associated unary language is
\[
 L^{\mathcal U}
 =\bigcup_{n\geq 1}L_n^{\mathcal U}
 =\{0^n:b_n=1\}.
\]
Thus deciding whether \(0^n\in L^{\mathcal U}\) means determining
whether \(U_n\) is the identity or a reflection.

\begin{theorem}[Simultaneous separation]
\label{thm:all-constants}
There exists a quantum oracle
\(\cU\) such that
\[
 L^{\cU}\in\textsf{UniqueQMA}^{\cU}
 \qquad\text{and}\qquad
 L^{\cU}\notin\QCIP^{\cU}.
\]
In particular, the same oracle satisfies
\(\textsf{UniqueQMA}^{\cU}\not\subseteq\QCIP[r]^{\cU}\) for every fixed integer
\(r\geq1\).
\end{theorem}

\begin{proof}
First, we show that \(L^{\cU}\in\textsf{UniqueQMA}^{\cU}\) for every oracle of this form.
Reject non-unary inputs. On input \(0^n\), the verifier receives an
\(n\)-qubit witness \(\ket{\phi}\), prepares a control qubit in
\(\ket{+}\), applies controlled-\(U_n\), applies a Hadamard gate to
the control, and accepts exactly when the outcome is \(1\).
If \(U_n=R_{\psi_n}\), the witness \(\ket{\psi_n}\) is accepted with
probability one. If \(U_n=\id\), every witness is rejected.

Enumerate all polynomial-time uniform quantum verifiers with classical
communication as \(V_1,V_2,\ldots\), using a fixed
finite universal gate set. Let \(q_i(n)\) and \(B_i(n)\) be polynomial
bounds on the oracle queries and total prover-to-verifier communication of
\(V_i\). This enumeration includes every verifier with polynomially many
messages.

At stage \(i\), choose a fresh input length \(n_i\), larger than every
length used earlier, such that the right-hand side of
Eq.~\eqref{eq:coarse-value}, with \(q=q_i(n_i)\), \(B=B_i(n_i)\), and
\(D=2^{n_i}\), is less than \(1/3\). This is possible because \(q_i\)
and \(B_i\) are polynomials whereas \(2^{n_i}\) is exponential.

Consider \(V_i\) on input \(0^{n_i}\), and let
\(v_0=\val_{\id}(V_i)\) when the length-\(n_i\) component is the
identity. There are two cases.
\begin{itemize}
\item If \(v_0>1/3\), set \(U_{n_i}=\id\). Then
      \(0^{n_i}\notin L^{\cU}\), but \(V_i\) violates soundness.
\item If \(v_0\leq1/3\), Theorem~\ref{thm:value} gives
      \[
       \E_\Psi[\val_{R_\Psi}(V_i)]
       \leq v_0+\E_\Psi[|\val_{R_\Psi}(V_i)-v_0|]
       <\frac23,
      \]
      where \(\Psi\) is Haar distributed on \(n_i\) qubits.
      Hence some \(\ket{\psi_{n_i}}\) satisfies
      \(\val_{R_{\psi_{n_i}}}(V_i)<2/3\). Set
      \(U_{n_i}=R_{\psi_{n_i}}\). Then \(0^{n_i}\in L^{\cU}\), but
      \(V_i\) violates completeness.
\end{itemize}

Set all unused components to the identity. The length-respecting restriction
on verifier oracle access ensures that later stages do not change
\(V_i\)'s behavior on \(0^{n_i}\). The value \(\val\) already maximizes
over every deterministic prover, so future oracle components cannot give
the prover an additional advantage omitted by the bound. Thus every
polynomial-round $\QCIP$ verifier fails on some input, while the
same language lies in \(\textsf{UniqueQMA}^{\cU}\).
\end{proof}

\section{An Average-Case \(\QMA\)--\(\QCIP[2]\) Gap Problem}
\label{sec:gap}

The unary language used in Section~\ref{sec:separation} gives a worst-case
separation, but it does not provide efficiently samplable instances and
witnesses.  We now give a sampler-augmented oracle satisfying the
average-case definition of Barhoush, Morimae, and Movassagh
~\cite{BarhoushMorimaeMovassagh2026}.  The construction and the seven-hybrid
proof follow the \(\QMA\)--\(\QCMA\) gap oracle of Barhoush, Morimae,
Nishimaki, and Yamakawa~\cite{BarhoushMorimaeNishimakiYamakawa2026}.
The new ingredient is the threshold-coupling argument inside
the proof of Theorem~\ref{thm:gap}, which controls adaptive queries to a
\(\QCIP[2]\) decision oracle.

\subsection{The relativized gap definition}

We use explicit descriptions over a fixed finite universal
gate set.  A query to a \(\QCIP[2]\) decision oracle is a string
\(z=(V_1,V_2,x)\).  On input \(x\), the first verifier circuit produces a
classical challenge \(a\) and a private quantum state that it passes to the  second verifier.  The prover
returns a classical answer depending on \(a\), and the second verifier
circuit decides whether to accept.  Let
\(\val_{\cT}(z)\) be the largest acceptance probability over all deterministic
classical provers when both verifier circuits have access to the quantum
oracle \(\cT\).  

We use the convention in \cite{BarhoushMorimaeMovassagh2026} to define a $\QCIP[2]$ oracle. Note that for a $\QCIP[2]$ language some inputs do not belong to yes or no instances, so it is not clear how the oracle should respond to such inputs. We deal with this as in \cite{BarhoushMorimaeMovassagh2026} by defining a whole class of oracles where any function can belong to the class as long as it accepts yes instances and rejects no instances. Then, any adversary that uses a $\QCIP[2]$ oracle to solve some task should be able to succeed regardless of which function in the class it is given oracle access to. We call a function in this class a completion of the $\QCIP[2]$ oracle. 

\begin{definition}[Relativized two-message decision oracle]
\label{def:qcip-decision}
The set \(\mathcal O_{\QCIP[2]}^{\cT}\) consists of all total Boolean
functions \(\cO\) satisfying
\[
 \cO(z)=1\quad\text{if }\val_{\cT}(z)\geq\frac23,
 \qquad
 \cO(z)=0\quad\text{if }\val_{\cT}(z)\leq\frac13.
\]
When the value lies strictly between \(1/3\) and \(2/3\), the answer is
unrestricted.  Oracle calls are classical.  In this oracle construction, a completion is an arbitrary total
Boolean function satisfying these two constraints.
\end{definition}

The quantifier over completions prevents an algorithm from relying on
specially chosen answers to ambiguous queries.

\begin{definition}[Relativized \textsf{QMA}--\textsf{QCIP}[2] gap problem]
\label{def:gap}
Let \(G^{\cT}\in\QMA^{\cT}\), let \(Q^{\cT}\) be its verifier, and let
\(\mathcal Y_n\) and \(\mathcal N_n\) be distributions supported on
\(G^{\cT}\cap\{0,1\}^n\) and
\(\{0,1\}^n\setminus G^{\cT}\), respectively.  The pair
\((\mathcal Y,\mathcal N)\) is a gap problem if:
\begin{enumerate}
\item for every polynomial \(t\), a \textsf{QPT} oracle algorithm on input
      \((1^n,1^{t(n)})\) samples
      \((x,\ket{w_x}^{\otimes t(n)})\), where
      \(x\) has distribution \(\mathcal Y_n\) and \(\ket{w_x}\) is accepted
      by \(Q^{\cT}\) with probability at least \(2/3\);
\item a possibly inefficient oracle algorithm samples \(\mathcal N_n\); and
\item for every polynomial \(s\) and every \textsf{QPT} \(A\), there is a completion
      \(\cO\in\mathcal O_{\QCIP[2]}^{\cT}\) such that
      \[
      \begin{aligned}
       &\Big|\Pr_{X_1,\ldots,X_s\leftarrow\mathcal Y_n}
          [A^{\cT,\cO}(X_1,\ldots,X_s)=1]\\
       &\quad-\Pr_{X_1,\ldots,X_s\leftarrow\mathcal N_n}
          [A^{\cT,\cO}(X_1,\ldots,X_s)=1]\Big|\\
       &\qquad\leq\negl(n).
      \end{aligned}
      \]
\end{enumerate}
The distinguisher has quantum query access to \(\cT\) and classical access
to \(\cO\).
\end{definition}

\subsection{The random-subset oracle}

Fix \(n\geq4\), let \(N=2^n\), and put
\[
             K=\left\lfloor 2^{n/2}\right\rfloor.
\]
Choose \(L_n\) uniformly among all \(K\)-element subsets of
\(\{0,1\}^n\).  Independently choose an \(n\)-qubit Haar-random state
\(\ket{\psi_x}\) for every \(x\in\{0,1\}^n\).

The oracle has one testing component and two sampling components.  The
testing component \(W_n\) has a control qubit, an \(n\)-bit label register,
and an \(n\)-qubit state register.  On computational-basis control and label
registers, it acts as
\[
 W_n\ket{c,x}\ket{\phi}
  =\ket{c,x}\,
    R_{\psi_x}^{\,c\,\mathbf 1[x\in L_n]}\ket{\phi}.
 \tag{8}\label{eq:indexed-test}
\]
Thus a controlled query at \(x\in L_n\) applies the reflection about
\(\ket{\psi_x}\), while a query at \(x\notin L_n\) applies the identity.

For a unit vector \(\ket a\) orthogonal to \(\ket0\), write
\[
 S(a)=\ket a\!\bra0+\ket0\!\bra a+\id-\proj0-\proj a.
 \tag{9}\label{eq:swap-unitary}
\]
This unitary swaps \(\ket0\) and \(\ket a\) and fixes their common
orthogonal complement.  For every integer \(t\geq1\), define
\[
 \ket{Y_{n,t}}
   =\frac1{\sqrt K}\sum_{x\in L_n}
       \ket{x}\ket{\psi_x}^{\otimes t}\ket1,
 \qquad
 \ket{N_n}
   =\frac1{\sqrt{N-K}}\sum_{x\notin L_n}\ket{x}\ket1.
 \tag{10}\label{eq:sampler-states}
\]
The oracle contains the blocks
\[
       \mathsf{GenYes}_{n,t}=S(Y_{n,t})
       \quad\text{for every }t\geq1,
       \qquad
       \mathsf{GenNo}_n=S(N_n).
\]
The integer \(t\) is a classical tag.  Equivalently, all these blocks form
one tagged block-diagonal unitary.  A polynomial-time machine can address
only tags and registers of polynomial length.

Let \(\cT\) denote the collection of these independently sampled components
over all input lengths, and define
\[
             G^{\cT}=\bigcup_{n\geq4}L_n.
\]
Membership at the finitely many smaller lengths is defined arbitrarily.

\begin{lemma}[Verification and sampling]
\label{lem:gap-sampling}
For every realization of \(\cT\), the language \(G^{\cT}\) belongs to
\(\QMA^{\cT}\) with perfect completeness and perfect soundness.  Moreover, the
uniform distributions on \(L_n\) and on its complement satisfy the first two
conditions of Definition~\ref{def:gap}.
\end{lemma}

\begin{proof}
On input \(x\), the \textsf{QMA} verifier performs the Hadamard test for the controlled
testing block \(W_n\).  If \(x\in L_n\), the witness
\(\ket{\psi_x}\) is a \(-1\) eigenvector and is accepted with probability
one.  If \(x\notin L_n\), the addressed block is the identity for every
witness, so the verifier rejects with probability one.

On input \((1^n,1^t)\), the yes sampler applies
\(\mathsf{GenYes}_{n,t}\) to \(\ket0\), measures the label register, and
returns the remaining witness registers.  Its output is a uniform
\(x\in L_n\) together with \(\ket{\psi_x}^{\otimes t}\).
Applying \(\mathsf{GenNo}_n\) to \(\ket0\) and measuring the label gives a
uniform element of the complement.  Each sampler makes one oracle query.
\end{proof}

\subsection{Stability of the oracle hybrids}

The seven-hybrid proof changes one sampler term or one reflection at a time.
The sampler transitions use the elementary perturbation estimate below.  The
reflection transitions will be handled directly inside the seven-hybrid proof
by applying Theorem~\ref{thm:value} to a suitable meta-verifier.
Throughout this subsection, ``exponentially small'' means at most
\(2^{-cn}\) for some constant \(c>0\) and all sufficiently large \(n\).
The constant and the starting point may depend on the fixed algorithm and
its polynomial resource bounds.  We freely absorb polynomial factors into
this notation.

\begin{lemma}[Puncturing a coherent sampler]
\label{lem:sampler-perturbation}
Let \(\ket A=M^{-1/2}\sum_{j=1}^M\ket{a_j}\), where the
\(\ket{a_j}\)'s are orthonormal, and let
\[
       \ket{A^{-i}}
       =\frac1{\sqrt{M-1}}\sum_{j\neq i}\ket{a_j}.
\]
Then
\[
 \norm{S(A)-S(A^{-i})}=O(M^{-1/2}).
\]
The same estimate holds when one orthogonal term is added to an
\(M\)-term superposition.
\end{lemma}

\begin{proof}
The two target states have inner product \(\sqrt{(M-1)/M}\), and hence
\[
 \norm{\ket A-\ket{A^{-i}}}^2
 =2-2\sqrt{1-\frac1M}
 =O(M^{-1}).
\]
Expanding the definition of \(S\) in Eq.~\eqref{eq:swap-unitary} and using
\(\norm{\proj a-\proj b}\leq
2\norm{\ket a-\ket b}\) gives the claimed operator-norm bound.
The calculation for adding a term is identical.  For every fixed
deterministic prover, the ordinary query hybrid changes the acceptance probability by at
most twice the number of queries times the operator-norm change.  Taking the
maximum over deterministic provers preserves this bound.

For the tagged family \(\{\mathsf{GenYes}_{n,t}\}_t\), puncturing the same
label changes every tagged block by the same \(O(M^{-1/2})\) amount.  The norm of
the tagged direct-sum difference is
the maximum of the block norms, so the estimate is uniform in \(t\).
\end{proof}

We now choose the answers on ambiguous \(\QCIP[2]\) queries.  For every
binary string \(z\), independently sample
\[
                   \Theta_z\leftarrow[2/5,3/5]
\]
once and keep it fixed.  Write \(\Theta=(\Theta_z)_z\) for this threshold
table.  Given a base oracle \(\cT\), define
\[
 \cO_{\cT,\Theta}(z)
        =\mathbf 1[\val_{\cT}(z)\geq\Theta_z].
 \tag{11}\label{eq:threshold-completion}
\]
The random table is only a probabilistic-method device.  Once it is fixed,
the completion is deterministic and every repeated query receives the same
answer.
Every such function is a valid deterministic completion: values at least
\(2/3\) are always accepted and values at most \(1/3\) are always rejected.
If two base oracles give values \(v\) and \(v'\) to the same description,
then, using the same threshold in both worlds,
\[
 \Pr_\Theta[
   \mathbf 1[v\geq\Theta_z]\neq
   \mathbf 1[v'\geq\Theta_z]]
 \leq5|v-v'|.
 \tag{12}\label{eq:threshold-change}
\]
Indeed, an answer differs only when the uniform threshold falls between
\(v\) and \(v'\), and the interval \([2/5,3/5]\) has length \(1/5\).

\subsection{The gap problem}

We follow the seven-hybrid argument of
Barhoush, Morimae, Nishimaki, and Yamakawa
\cite[Theorem~8]{BarhoushMorimaeNishimakiYamakawa2026}, using
Theorem~\ref{thm:value} to handle classical interaction.

\begin{theorem}[Oracle gap problem]
\label{thm:gap}
For some realization \(\cT\) of the construction above, the uniform
distributions on \(L_n\) and its complement form a
\(\QMA\)--\(\QCIP[2]\) gap problem against uniform \textsf{QPT} distinguishers.
\end{theorem}

\begin{proof}
\textsf{QMA} membership and the sampling requirements follow from
Lemma~\ref{lem:gap-sampling}. We prove indistinguishability first for one
input label. Fix a uniform \textsf{QPT} algorithm \(A\), and let \(Q=Q(n)\) be a
polynomial bounding its queries and the description lengths and resources
of the protocols it can submit to the completion oracle, under the stated
encoding conventions.

Sample the base oracle \(\cT\), a uniform \(y\in L_n\), and a uniform
\(\bar y\notin L_n\). Independently sample the threshold table \(\Theta\)
and use it throughout. Whenever a hybrid changes the base oracle, its
completion is changed accordingly, using this same table.

Consider the following hybrids, each of which runs \(A\) and outputs its
result.
\begin{description}
\item[\(H_0\).] Use the original oracle and give \(A\) the input \(y\).
\item[\(H_1\).] Remove the term indexed by \(y\) from every
      \(\mathsf{GenYes}_{n,t}\).
\item[\(H_2\).] Replace the reflection at address \(y\) by the identity.
\item[\(H_3\).] Add \(y\) to \(\mathsf{GenNo}_n\).
\item[\(H_4\).] Keep this modified oracle, but replace the input \(y\)
      by \(\bar y\).
\item[\(H_5\).] Remove \(y\) from the no sampler.
\item[\(H_6\).] Restore the reflection at address \(y\).
\item[\(H_7\).] Restore the term indexed by \(y\) in every yes sampler.
\end{description}
Each hybrid changes only the operation specified, leaving everything else
as in the preceding hybrid. Thus \(H_7\) uses the original oracle and a
uniform no-instance.

We first explain how to account for the completion oracle when comparing
two neighboring base oracles \(\cT\) and \(\cT'\). Change the completion
while keeping the base oracle fixed, then change the base oracle while
keeping the completion fixed. For the first change, couple the executions
until their first different completion answer. On a new query \(z\), the
threshold \(\Theta_z\) is independent of the preceding history, and
Eq.~\eqref{eq:threshold-change} bounds the disagreement probability by
\(5|\val_{\cT}(z)-\val_{\cT'}(z)|\). A repeated query cannot cause
the first disagreement. Hence a union bound over the at most \(Q\)
queries bounds this change by
\[
 5Q\max_z|\val_{\cT}(z)-\val_{\cT'}(z)|,
\]
where the maximum is over the descriptions allowed by the resource bound.
This maximum also covers adaptively chosen descriptions.

\emph{Sampler changes.}
The transitions \(H_0\) to \(H_1\), \(H_2\) to \(H_3\), \(H_4\) to
\(H_5\), and \(H_6\) to \(H_7\) each add or remove one sampler term.
By Lemma~\ref{lem:sampler-perturbation}, the operator-norm change is
exponentially small, uniformly over all tags. The ordinary query hybrid
therefore bounds both the change in every protocol value and the direct
change seen by \(A\). The preceding completion bound preserves exponential
smallness.

\emph{Reflection changes.}
For \(H_1\) to \(H_2\) and \(H_5\) to \(H_6\), the yes samplers no
longer contain \(y\). Fix the labels and all oracle data other than
\(\psi_y\). The remaining background is independent of \(\psi_y\).
Apply Theorem~\ref{thm:value} to a verifier that first receives a protocol
description \(z\) from its prover and then executes that protocol.
The description and subsequent prover messages have polynomial total
length, and the verifier makes polynomially many queries. Its prover may
choose both \(z\) and the later replies as a function of \(\psi_y\).
For each \(z\), the difference between optimal values is at most the
maximum acceptance-probability difference for the same deterministic prover.
The bound on this maximum in Theorem~\ref{thm:value} therefore gives
\[
 \E_{\psi_y}\!\left[
   \max_z|\val_{\cT}(z)-\val_{\cT'}(z)|
 \right]\leq 2^{-\Omega(n)}.
\]
Thus the completion change is exponentially small. For the direct oracle
change, reverse the comparison if necessary and keep the completion
belonging to the world in which the reflection has been removed.
This completion and the punctured samplers are independent
of \(\psi_y\), so Theorem~\ref{thm:value} with no prover messages also
bounds the direct change by \(2^{-\Omega(n)}\).

\emph{The middle switch.}
Put \(S=L_n\setminus\{y\}\). Conditioned on \(S\), the ordered pair
\((y,\bar y)\) is uniform over distinct elements of \(S^c\).
In \(H_3\) and \(H_4\), the testing and yes-sampling components use
exactly the labels in \(S\), and the no sampler is uniform over \(S^c\).
The modified oracle and its completion therefore do not depend on which
element of the pair was originally in \(L_n\). Exchanging \(y\) and
\(\bar y\) leaves their joint distribution unchanged, so \(H_3\) and
\(H_4\) have identical output distributions.

The triangle inequality now gives
\begin{equation}
 |\Pr[H_0=1]-\Pr[H_7=1]|\leq 2^{-\Omega(n)}.
 \label{eq:gap-single-sample}
\end{equation}
These probabilities include the choice of the base oracle and threshold
table. To obtain a fixed realization, write \(p_Y\) and \(p_N\) for
\(A\)'s yes- and no-input acceptance probabilities for fixed
\((\cT,\Theta)\). Apply Eq.~\eqref{eq:gap-single-sample} to an algorithm
that runs \(A\) on its challenge and independently on an internally sampled
yes-label, accepting only if both runs accept. Repeat with an internally
sampled no-label. This bounds the absolute values of
\(\E[p_Y^2-p_Yp_N]\) and \(\E[p_Yp_N-p_N^2]\) by
\(2^{-\Omega(n)}\). Subtracting gives
\[
 \E_{\cT,\Theta}[(p_Y-p_N)^2]\leq 2^{-\Omega(n)}.
\]
Markov's inequality and the Borel--Cantelli lemma imply that
\(|p_Y-p_N|\) is negligible with probability one over \((\cT,\Theta)\).
There are countably many uniform \textsf{QPT} algorithms, so we can fix a realization
for which this holds for every such algorithm.

Finally, polynomially many independent samples follow by the usual hybrid
over the input coordinates. A single-challenge distinguisher chooses a
uniformly random coordinate for its challenge and samples the preceding
coordinates from the yes distribution and the following coordinates from
the no distribution, using the public sampling oracles. Its advantage is
the original distinguishing advantage divided by the number of samples.
Thus a non-negligible advantage on polynomially many samples would
contradict the single-sample bound for the fixed realization. This proves
the remaining requirement of Definition~\ref{def:gap}.
\end{proof}

\end{document}